\documentclass{amsart}
\usepackage{amsmath,amssymb,amsfonts,amsthm}
\usepackage{amssymb}
\usepackage{pgf}
\usepackage{todonotes}
\usepackage{epsfig}
\usepackage{mathrsfs}
\usepackage{verbatim}
\usepackage[gen]{eurosym} 
\usepackage{subfig}

\newtheorem{theorem}{Theorem}[section]

\newtheorem{proposition}[theorem]{Proposition}

\theoremstyle{definition}

\newtheorem{example}[theorem]{Example}

\theoremstyle{remark}

\newcommand{\n}{\!\!\ \nn \!\!\ }

\def\nn{\mathrel{%
    \mathchoice{\n}{\n}{\scriptsize\n}{\tiny\n}%
}}
\def\n {{%
    \setbox0\hbox{ \textsf{n}\!\  }%
    \rlap{\hbox to \wd0{\hss \!\ --\hss}}\box0
}}

\numberwithin{equation}{section}

\newcommand{\R}{\mathbb{R}}

\renewcommand{\P}{\mathbf{P}}

\newcommand{\ua}{\uparrow}
\newcommand{\da}{\downarrow}
\begin{document}

\title[Option pricing simplified]{Option pricing: A yet simpler approach}

\author{Jarno Talponen \and Minna Turunen}
\address{University of Eastern Finland, Department of Physics and Mathematics, Box 111, FI-80101 Joensuu, Finland}
\email{talponen@iki.fi \and minna.turunen@iki.fi}


\date{\today}

\begin{abstract}
We provide a lean, non-technical exposition on the pricing of path-dependent and European-style derivatives in the Cox-Ross-Rubinstein (CRR) pricing model. The main tool used in the paper for cleaning up the reasoning is applying static hedging arguments.

This can be accomplished by taking various routes through some auxiliary considerations, namely Arrow-Debreu securities, digital options or backward random processes. In the last case the CRR model is extended to an infinite state space which leads to an interesting new phenomenon not present in the classical CRR model. 

At the end we discuss the paradox involving the drift parameter $\mu$ in the BSM model pricing. We provide sensitivity analysis and the speed of converge for the asymptotically
vanishing drift.
\end{abstract}

\maketitle

\section{Introduction}

In this paper we provide a transparent and financially tractable approach to verifying financial derivatives 
pricing formulas in a lattice model. 

The derivatives pricing model originated in the seminal papers of Black and Scholes (1973) and Merton (1973) (BSM) is the corner stone of modern derivatives pricing. Understanding their approach fully requires some rather involved mathematical machinery. In an attempt to alleviate this burden, Cox, Ross and Rubinstein (1979) (CRR) introduced a lattice model which approximates the BSM prices with a very rapid rate of convergence as the number of time steps grows (see e.g. Leisen and Reiner 1996). Understanding the CRR model requires considerably less mathematical sophistication than the BSM model.

The celebrated Cox-Ross-Rubinstein binomial option pricing formula states that the price of an option is
\begin{equation}\label{eq: CRRformula}
C_f(0)=\frac{1}{(1+r)^T}\sum_{x=0}^T f\left(S_0(1+u)^x(1+d)^{T-x}\right)\left(\binom{T}{x}q^x(1-q)^{T-x}\right).
\end{equation}
where $f$ denotes the payoff of the European style derivative at maturity, $T$ denotes the time steps to maturity and $r$  is the risk-free interest rate corresponding to each time step, and $q$ can be easily calculated from the parameters of the model.  

There is a vast literature of lattice models in finance. Lattice models inspired by the CRR model have been applied e.g. to financial derivatives pricing (Babbs 2000), state price density estimation by implied trees (Rubinstein 1994), real options valuation (Nembhard et al. 2002, 2003), investment science, hybrid securities (Das and Sundaram 2007, Gamba and Trigeorgis 2007), and term structure models (Heath et al. 1990). Here we also study the implications of extending the state space in the CRR binomial model. Previously, the CRR model has been extended in various manners, for example, by Boyle (1988) to value options with several state variables, by Broadie and Detemple (1996) to value American style options, or by Hull and White (1993) and Kascheev (2000) to value path-dependent options.

The CRR model is easy to grasp in principle, and thus the apparently more complicated BSM model can be understood as well by extension, since it can be seen as an asymptotic limit of CRR models. Unfortunately, the crucial step in the CRR paper, where their main pricing formula is actually justified, is swept under the rug; after discussing the first two steps Cox et al. state that they ``now have a recursive procedure for finding the value of a call with any number of periods to go'' (1979, p. 238)\footnote{Hull (2015) takes essentially the same approach. 
F\"{o}llmer and Schied (2011) develop rigorously the machinery in Ch.5 with martingales.}. The required backward substitution calculations become lengthy, especially for a general path-dependent payoff $f$, even if the idea is simple in principle. Although the CRR model was introduced as a simplified version of the BSM model, and well succeeds in that, some steps of the calculations remain not that transparent at first glance, say, to a student.

We have not been able to find a \emph{lean} argument for the CRR pricing formula \eqref{eq: CRRformula} in the quantitative finance literature. The rigorous arguments there become usually somewhat complicated, they require probability-theory, e.g. martingales, and the financial intuition may easily be lost in the details.

Consequently, there is a rough passage starting with rudimentary considerations to the financial understanding of the BSM model. Our aim here is to provide a fix to this `gap' in the story, essentially by using static hedging arguments. Also we hope that our method makes the CRR model somewhat more approachable, especially from a pedagogic point of view. Understanding our approach does not require such an extensive knowledge of probability theory.

Thus, the main contribution of this paper is not a novel result but rather we will give a lean, financially 
oriented argument for both the classical European-style derivative pricing formula and the general path-dependent option price formula in the CRR model. We will apply Arrow-Debreu type securities (Arrow and Debreu 1954) and digital options as convenient intermediate notions towards verifying the CRR pricing formulas. These securities are financially well motivated since they can be considered as natural building blocks for other financial derivatives, and in our case, especially options. The Arrow-Debreu securities are not actively traded in the real market
but digital options are. Even if the Arrow-Debreu securities are not traded by themselves, traded structured products plausibly consist of such securites. Thus these securities appear more tractable than their alternative, risk-neutral probability densities.


This paper is organized as follows. First, we recall the binomial model and explain various types of atomic building blocks in our model. We show how the prices of Arrow-Debreu (AD) securities, that is, kind of elementary options on 
particular trajectories of the underlying security prices, arise in a rather simple way. Then we obtain the path-dependent derivative prices by suitably aggregating these AD securities. It turns out that the classical European style derivatives pricing formula 
follows easily by aggregating binary options. These, in turn, are aggregated from AD securities, or, alternatively, can be priced by means of a simple backward random walk in an extended state space. It turns out that in the case of an extended infinite state space the discounted value processes exhibit an interesting aggregate time invariance, not present in the standard binomial model. At the end of the paper we discuss the irrelevance of the trend parameter $\mu$
in the BSM pricing which is a bit of a paradox. 

We have made an effort to explain carefully the strategy behind the pricing of general financial derivatives in the CRR
model without resorting to unnecessary technical machinery. Instead of fictitious risk-neutral probabilities we mainly 
consider financially tractable elementary securities. In particular, Section \ref{Subsection:path-dependentAD} hopefully serves as an `executive summary' on the CRR pricing principles which essentially entail all the financial reasoning behind the BSM model.

\subsection{Preliminaries}

Although we do not assume knowledge of lattice models in-depth, we expect some familiarity 
with related financial literature. For a suitable background information see, for example, the monographs by Copeland and Weston (1992), Luenberger (1998), or Hull (2015), cf.  F{\"o}llmer and Schied (2011), van der Hoek and Elliot (2006). 
\subsubsection{Some notations} The indicator function becomes a very useful notion here, flexibly defined, 
\[1_{C}\] 
means a function which has value $1$ if the subscript condition $C$ is valid and otherwise has value $0$.
The underlying asset's price at time $t$ is denoted by $S_t$. We denote by $f$ the payoff of some financial derivative of interest. It encodes the information of the payment of the derivative. For instance, a European-style call has a time $T$ payoff of the form
\[f(\omega)=\max(S_T(\omega) - K , 0)\]
and a barrier put option payoff may have the form
\[f(\omega) = 1_{\{\max_{0\leq t\leq T} S_t(\omega) \geq L\}} \max(K-S_T(\omega) , 0).\]

\subsection{The basic binomial model}

Although we will not require much probability theory here, let us just mention that our technical setup is a binomial model 
$(B,S,\Omega,\mathcal{F},\mathbb{F},\mathbf{P})$ where 
\[\Omega=\{\omega=(\theta_1 , \theta_2 , \ldots ,\theta_T)\colon \theta_t =0,1\}\] 
is the sample space, $\mathcal{F}$ is a $\sigma$-algebra representing the set of events (here we can choose $\mathcal{F}$ to be the collection of all subsets of $\Omega$), $\mathbb{F}$ is a filtration and $\mathbf{P}$ is a probability measure. 

The (nominal) value of the underlying asset at time $t=0,1,\ldots , T$ is denoted by $S_t$. 
In a binomial model, $S:\Omega\times\{1,\ldots,T\} \to \R$ is a random variable with two possible outcomes `up' ($1$) and `down' ($0$), at each step, so that the possible nominal values of $S_{t+1}$ are $S_t (1+u)$ and $S_t (1+d)$. It is assumed that $d<0<r<u$, where $r$ is a constant short interest rate and 
\[(1+d)(1+u)=1,\]
so that the binomial tree is recombining. 
The reasonable choices of $d,r$ and $u$ depend on the length of the time steps. The probabilities are defined as follows:
\[S_{t+1} = S_t (1+d+\theta_{t+1}(u-d))\]
where $\theta_i$, $1\leq i\leq T$, are i.i.d with $\P(\theta_i =1)=p$ and  $\P(\theta_i =0)=1-p$ for some given $0<p<1$. The Figure \ref{Figure:Binomialmodel}  illustrate the binomial model in Log and real scale.

{
\centering
\begin{figure}[h!!]
     \centering
     \subfloat{\includegraphics[width=0.5\linewidth]{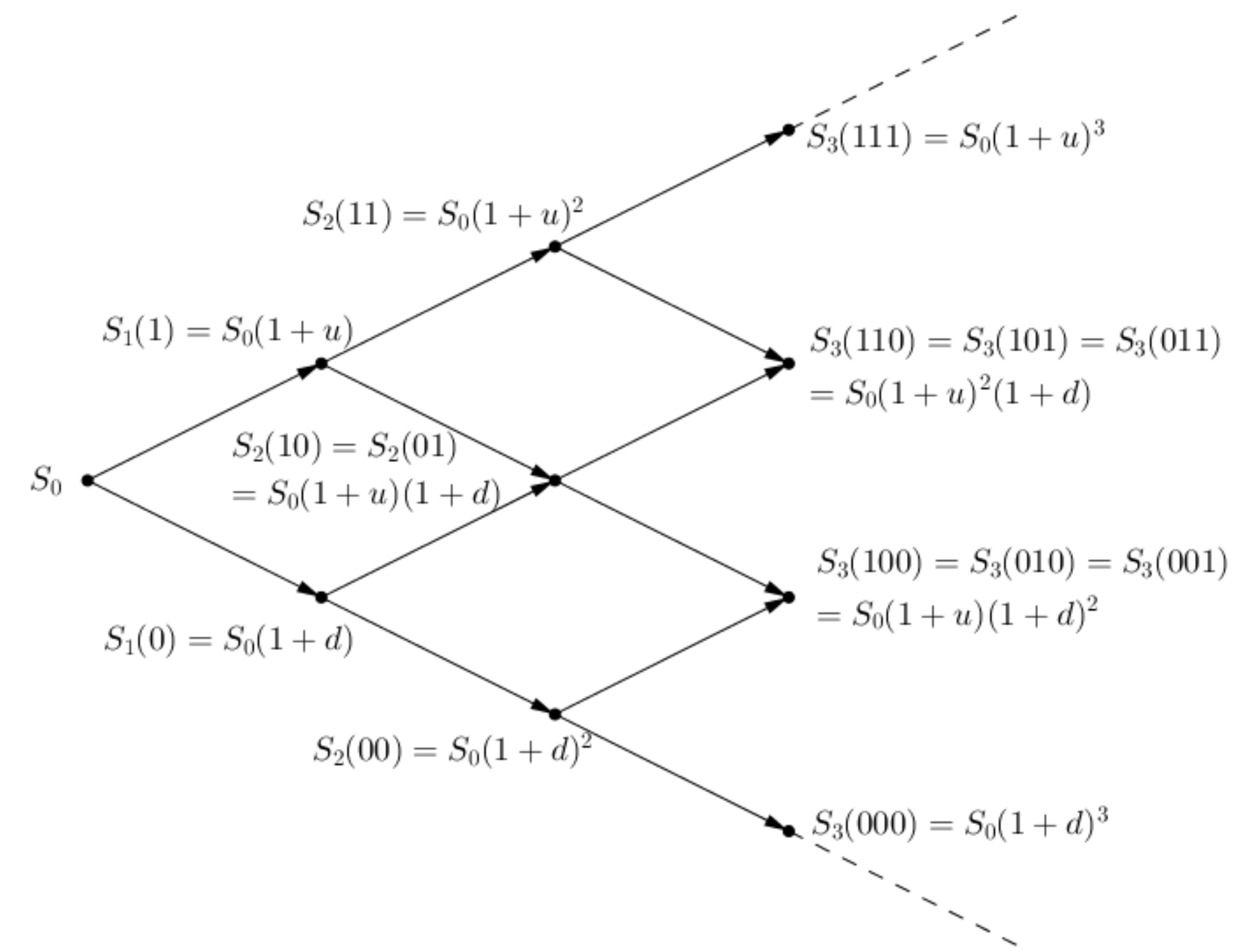}}
     \subfloat{\includegraphics[width=0.5\linewidth]{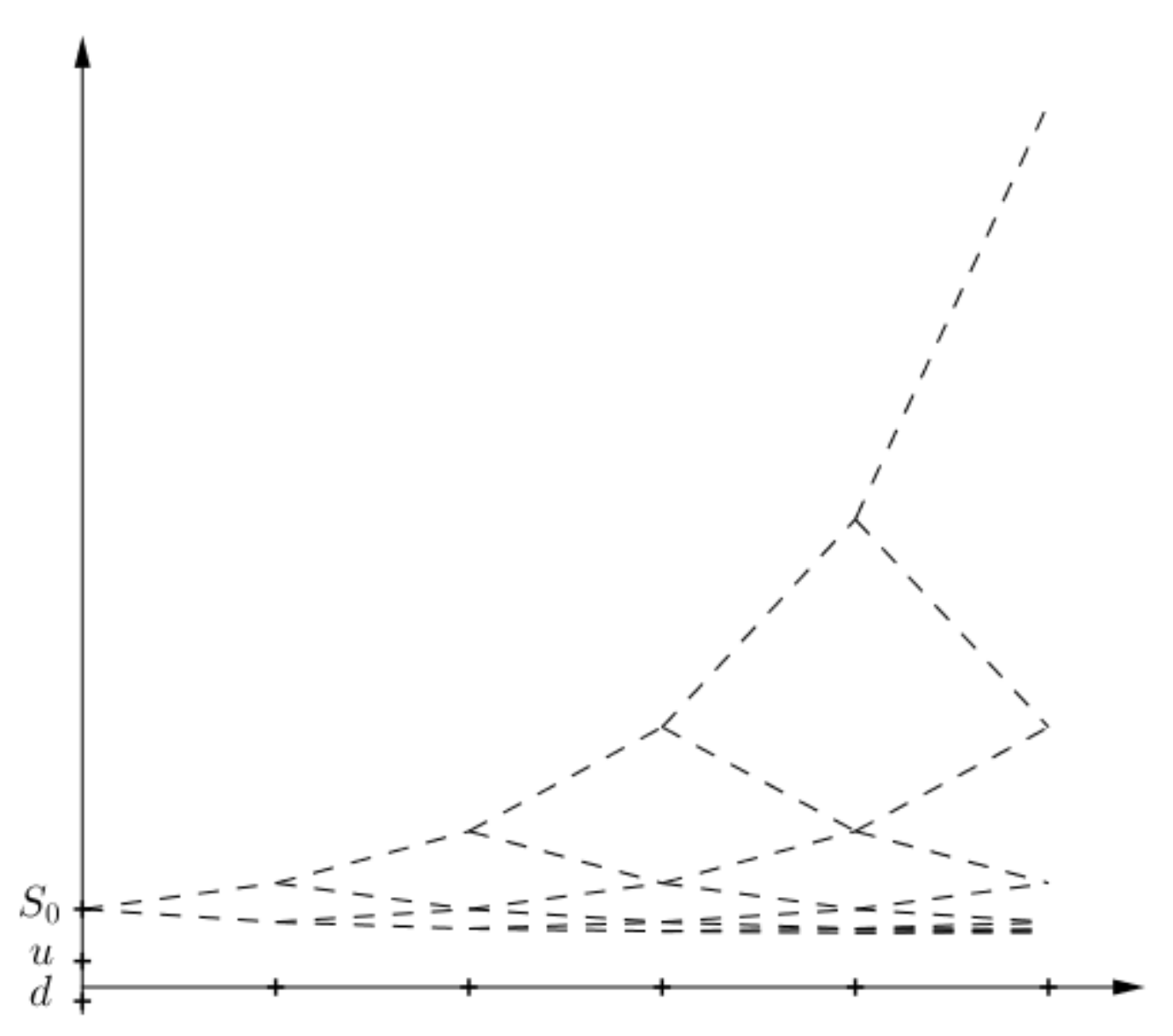}}
     \caption{The binomial model drawn in Log-scale and in real scale.}
     \label{Figure:Binomialmodel}

\end{figure}
}

The sample space essentially consists of all possible trajectories of $S$, see also Figure \ref{figure:valueprocess}.

As usual, $B_t$ denotes the riskless asset with the nominal value
\[B_t = \frac{(1+r)^{t}}{(1+r)^{T}}\]
expressed units of a given currency. This is a zero-coupon bond with face value $1$.

\subsection{Discounted model}
To simplify the arguments, it is customary in the quantitative finance literature `to pass on to a discounted model' where discounted prices appear in place of nominal prices. To perform this transition explicitly, we will do the book keeping in numeraire units. Expressing prices in numeraire terms is a bit like reporting inflation adjusted prices over a time span.
Our numeraire $\n$ incorporates the currency and discounting, and it \emph{depends on time} $t$ as follows:
\[\n \text{value\ of\ bond\ } B\ \text{at time}\ t \quad =\quad 1\ \n_{(T)},\]
or, in short,
\[B_t  \n_{(t)} =B_T \n_{(T)} =\n_{(T)},\]
for all $t$. The left hand numeraire corresponds to time $t$ and the right hand one corresponds to time $T$.
This leads to the following dimension analysis:
\[\frac{\n_{(t+1)}}{\n_{(t)}}=\frac{B_{t}}{B_{t+1} }=\frac{1}{1+r} .\]
That is, the net present value (NPV) of future certain cash flow $x \n$, considered as cash at present time, is 
$x  \n$. Hence the CRR formula can be expressed as
\begin{equation}
\label{eq:CRR1}
C_f(t) \n_{(t)} = \sum_{x=0}^{T-t}\left(\binom{T}{x}q^x(1-q)^{T-t-x}\right)  f\left(S_t (1+u)^x (1+d)^{T-t-x}\right) \n_{(T)} 
\end{equation}
where the left hand numeraire corresponds to cash at time $t$, whereas the right hand numeraire corresponds 
to future payoff at the maturity $T$. More generally, the time subscript corresponds to the value process time, 
so terminal payoff is always expressed in $\n_{(T)}$ and time $t$ value of any security in $\n_{(t)}$.
Following this convention we may suppress the times in subscripts and, for instance, the previous formula becomes simply
\[C_f(t) \n = \sum_{x=0}^{T-t}\left(\binom{T}{x}q^x(1-q)^{T-t-x}\right)  f\left(S_t (1+u)^x (1+d)^{T-t-x}\right) \n.\] 

\section{Streamlined argument for the CRR pricing formulas}\label{Section2}

The aim of this section is to explain the idea of CRR pricing in a transparent manner.

\subsection{Static hedging by Arrow-Debreu securities and digital options}
Static hedging (cf. Derman et al. 1995, Brown and Ross 1991) means synthesizing some required new securities
(or pricing existing ones) by running a buy-and-hold strategy on some existing securities.
The securities included long/short in the replicating portfolio are typically derivatives which are simpler than the 
new synthesized derivative security. If it is possible to construct a portfolio whose value at the maturity of the European style derivative
exactly matches the value of the derivative, then according to `no free lunch' principle, the initial price of the portfolio should
be the same as the price of the new derivative. Indeed, otherwise some very lucrative trading strategies arise where 
one can make money, essentially risk-free and from nothing. These are too good to be true and over some time they should cease to exist due to extensive arbitrage activity. We refer to this sort of economic reasoning as the 
\emph{static hedging principle}. 

We consider here two kinds of elementary derivatives, path-dependent ones, \emph{Arrow-Debreu securities}, and 
path-independent ones, namely \emph{degenerate digital options}. An Arrow-Debreu security's payoff is $1\n$ at the time of the maturity $T$ if the underlying asset evolution follows a given prescribed trajectory $\omega$, and is $0\n$ otherwise. Arrow-Debreu securities may be economically more tractable than risk-neutral probabilities. Neither of them 
are traded directly.

A (degenerate) digital option pays $1\n$ at maturity $T$ if the underlying asset hits a given prescribed `strike price' $K$
at time $T$, and pays $0\n$ otherwise. Digital options are traded and their prices can be estimated from European-style
option prices. It was shown by Breeden and Litzenberger already in (1978) that for plain vanilla calls and puts
there is an elegant model-free way to do this.

\subsection{AD securities in a $1$-step $2$-state model}
\label{Subsection:ADsecurities}
Let us consider Arrow-Debreu derivatives in a $1$-step model with times $0$ and $1$. The $t=0$ state is  
$S_0$ and $t=1$ time possible states are $S_1 = S_0 (1+u) ,\ S_0 (1+d)$.
The payoff functions are
 \[
 f_{AD\ua}=\mathbf{1}_{\{S_1=S_0(1+u)\}} \quad \text{and} \quad f_{AD\da}=\mathbf{1}_{\{S_1=S_0(1+d)\}}
 \]
which are known at time $t=1$.
In other words, the derivative ${AD_\ua}$ pays $1 \n$  when the value of the underlying asset goes up, and respectively, the derivative ${AD_\da}$ pays $1 \n$ when the value of the asset goes down. It turns out that the prices of 
these Arrow-Debreu derivatives at time $t=0$ are
\[
AD_\ua (0) \n=q \n:=\frac{r-d}{u-d} \n \quad \text{and} \quad AD_\da (0)\n=(1-q)\n = \frac{u-r}{u-d}\n.
\]
Note that one may statically hedge the risk-free zero-coupon bond $B$ with $\n$-unit face value by combining these $AD$ securities, since their total payoff is 
\[ f_{AD\ua}\n + f_{AD\da}\n =1\n ,\]
the payoff of the bond, at time $t=1$. Therefore it makes sense that 
\[AD_\ua (0) \n_{(0)} + AD_\da (0)\n_{(0)} = q\n_{(0)} + (1-q)\n_{(0)} =1\n_{(0)} =B_0 \n_{(0)} .\]

So, how to replicate an $AD$ security by a buy-and-hold strategy of assets and bonds where both 
long and short positions are available? To replicate the payoff of the $AD_\ua$ security 
we simply invest at $t=0$ a certain amount of numeraire in the stocks, $a S_0 \n$, and certain amount, $b B_0 \n$, in risk-free bonds. Recall that $B_1 \n_{(1)} = (1+r) B_0 \n_{(1)}$. The payoff replication conditions for ${AD_\ua}$ can be formalized as follows:
\[
  \begin{cases} 
   a S_0 (1+u)\n  + b (1+r) B_0 \n =1\n\\
   a S_0 (1+d)\n + b (1+r) B_0 \n=0\n
  \end{cases}
\]
Without loss of generality we may assume (by splitting assets or bundling them up) that $S_0=1$ above. 
Thus we get 
\[
  \begin{cases} 
   a (1+u) + b(1+r)  =1\\
   a (1+d) + b(1+r)  =0
  \end{cases}
\]
which can be solved by Gaussian elimination or inverting the coefficient matrix. 

However, there is a very natural financial approach to finding the right weights $a$ and $b$.
The variability of the portfolio consisting of stocks and bonds depends \emph{only} on the amount of 
stocks. Thus, $1-0 = a (u-d)$, so that $a=1/(u-d)$. Here $a>0$, since the portfolio, the $AD_\ua$ security and $S$ move in the same direction. Note that in the bearish scenario the time $t=1$ value 
of the bonds (shorted) should be the negative of the value of stocks in the portfolio,
\[b B_1= - a S_0 (1+d)\]
so
\[b =-aS_0 (1+d)/(1+r) = -(1+d)/((u-d)(1+r)),\] 
so that the nominal value of the portfolio is
\[AD_\ua (0)  =(a S_0 + bB_0) =\frac{(1+r)-(1+d)}{(u-d)(1+r)} =\frac{r-d}{(1+r)(u-d)}\]
and
\[AD_\ua (0)  \n_0 = q \n_1 .\]

Similarly we observe that in calculating the replication for $AD_{\da}$ security we must have
$a'=-1/(u-d)$ (same amount of absolute variation as above but this time contrary to the asset movement) and 
$a' (1+u) + b'(1+r)  =0$, thus 
\[b'=-a'(1+u)/(1+r)=(1+u)/((u-d)(1+r)),\] 
\[AD_\da (0)=a' + b'=\frac{-(1+r) + (1+u)}{(u-d)(1+r)}=\frac{u-r}{(u-d)(1+r)}\]
and 
\[AD_\da (0)  \n_0 = (1-q) \n_1 .\]
We may check by static hedging argument that the asset has the price assumed:
\begin{equation*}
\begin{split}
S_0 \n &=  \sum_{x=\da,\ua} f_S (x) AD_x (0)\n \\
&= (1+d) \frac{u-r}{(u-d)(1+r)}\n_{(0)}+ (1+u) \frac{r-d}{(1+r)(u-d)}\n_{(0)}\\
&=\frac{u-r+du-dr+r-d+ur-du}{(1+r)(u-d)}\n_{(0)}=\frac{(1+r)(u-d)}{(1+r)(u-d)}\n_{(0)}=1\n_{(0)}.
\end{split}
\end{equation*}

\subsection{CRR pricing: The path-dependent case simplified}
\label{Subsection:path-dependentAD}

Let us first discuss the pricing of path-dependent Arrow-Debreu derivatives. Suppose that we want to price a path-dependent $AD$ derivative that pays us $1\n$ if the evolution of the underlying asset follows exactly a given trajectory encoded in $\omega$ (a list of ups and downs), and $0\n$ if the asset's evolution diverts from this fixed trajectory at any time $t\leq T$. 

In pricing of this $AD$ derivative, we utilise the $1$-step $AD_{\ua}$ and $AD_{\da}$ derivatives presented in the previous section. As we recall, $AD_{\ua}$ is a derivative that costs $q\n$, and pays us $1\n$ if the value of the underlying asset goes up and $0\n$ otherwise.
Respectively, $AD_{\da}$ is a derivative with price $(1-q)\n$, and with a payoff $1\n$ if the value of the underlying goes down.

The idea of the pricing the $AD$ derivative is to construct a replicating portfolio from $AD_{\ua}$ and $AD_{\da}$ derivatives step-by-step, according to the trajectory related to the $AD$ derivative's payoff.
The construction proceeds from the time of the maturity to time $t=0$, and the idea of the construction is rather simple. Basically, at each time $t<T$, we consider the coordinate of the trajectory, $\omega_{t+1}$, that is, the movement of the underlying at that time. At time $t$
\begin{itemize}
\item  if $\omega_{t+1}=1$, i.e. the prescribed trajectory goes up, we synthesize a suitable number of shares of $AD_\ua$ derivatives; and
\item if $\omega_{t+1}=0$, i.e. the prescribed trajectory goes down, we synthesize a suitable number of shares of $AD_\da$ derivatives.
\end{itemize}
This hedge at time $t$ will provide us a suitable return at time $t+1$ so that we can perform the appropriate hedge at the following time steps as well. Repeating this step-by-step hedging strategy will provide us the static hedging portfolio, the discounted value of this portfolio, and, as the result, the discounted price of the $AD$ derivative.

Let us take an example by pricing the path-dependent $AD$ derivative and the fixed trajectory presented in Figure \ref{figure:ADderivative} below.

\begin{figure}[h!]
\begin{center}
\includegraphics[width=11.3cm]{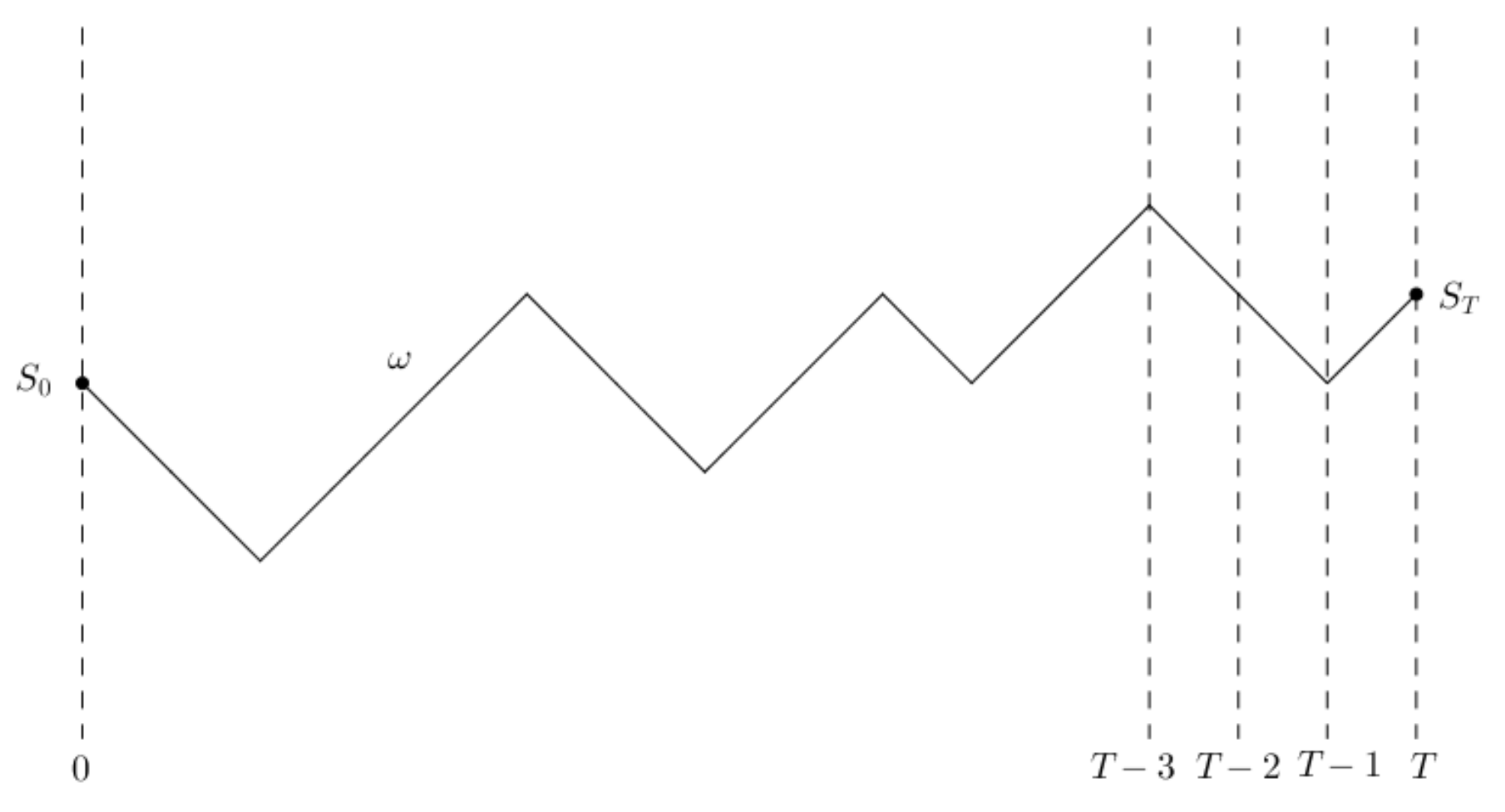}
\caption{The fixed path in the example of pricing a path-dependent $AD$ derivative in Log-scale.}
\label{figure:ADderivative}
\end{center}
\end{figure}

Let us begin the hedging by considering the situation right before the maturity, at time $T-1$. If we are not on the trajectory, then no wealth is required to cover the path-dependent $AD$, since it is worthless. So, let us assume we are on the trajectory. At time $T$ we want the hedging portfolio to pay us $1\n$ if the underlying stock has the fixed value $S_T(\omega)$. Since we know that we are on the trajectory, the stock satisfies $S_T=S_{T-1}(1+u)$, as in the Figure \ref{figure:ADderivative}. Hence, we can (super)hedge the derivative by buying an $AD_{\ua}$ derivative which costs us $q\n$.

Let us then consider the situation two periods before the maturity, at time $T-2$. The situation is almost the same as above but instead of getting $1\n$ from the hedging portfolio at the next period, we want the portfolio to pay us $q\n$. Indeed, with this amount of wealth we can run the previously described hedge at time $T-1$, so that it will, in turn, pay us $1\n$ at the time of the maturity. Again, let us assume that the stock now satisfies $S_{T-1}=S_{T-2}(1+d)$. Hence, we can hedge the derivative by buying $q$ shares of $AD_{\da}$ derivatives; these will pay us $1\n$ each, so that at time $T-1$ we will have $q\n$. Thus, at time $T-2$, the required wealth is $q(1-q)\n$.

With similar reasoning, considering the time $T-3$, we require the wealth to buy $q(1-q)$ shares of $AD_{\da}$ derivatives, i.e. at time $T-3$ we require the wealth $q(1-q)^2$, in order to enable the latter phases of the hedging strategy.

We can continue this backward recursion step-by-step so at time $t=0$ we will have the price of the $AD$ derivative, i.e. the amount of wealth required to initiate the strategy. The required initial wealth of the $AD$ derivative replication 
strategy is
\begin{equation}
\label{Eq:ADprice}
C_{\omega}(0) \n_{(0)} =q^x(1-q)^{T-x} \n_{(T)}
\end{equation}
where $x$ and $T-x$ denote the number of `ups' and `downs', respectively, in the fixed trajectory $\omega$, or equivalently, the number of phases where we use single-step $AD_{\ua}$ and $AD_{\da}$ derivatives, respectively. 

We need to bear in mind that if the value of the underlying asset diverts from the fixed path at any time $t\leq T$, the path-dependent $AD$ derivative is worthless, and therefore, the price of it is $0\n$, and also the hedging strategy ends there. On the other hand, if the evolution follows the fixed trajectory, then the hedging strategy returns $1\n$.
Consequently, the described hedging strategy yields exactly the same payoff as the path-dependent $AD$ derivative.
 
According to the static hedging principle we may construct any path-dependent derivative in the model by aggregating it as 
a suitable portfolio of $AD$ securities $C_\omega$. Namely, if the payoff involving trajectory $\omega$ is $f(\omega)$, then
we may accomplish this in the portfolio by including $f(\omega)$-many $AD$ securities $C_\omega$. Thus, for each $\omega$ the weight of an $AD$ security $C_\omega$ is $f(\omega)$. The price of the path-dependent option $C_f$ is then 
\begin{equation}\label{eq: pathdep}
C_f (0) \n_{(0)} = \sum_\omega f(\omega) C_\omega (0) \n_{(0)} = \sum_\omega f(\omega) q^{x(\omega)} 
(1-q)^{T-x(\omega)} \n_{(T)}.
\end{equation}

At the end of this section we will discuss pricing a \emph{path-independent} European call option using the hedging strategy described here, see Example \ref{example: }.

\subsubsection{The CRR pricing formula by considering combinations of digital options} 

Let us then discuss pricing a general European style derivative with a payoff function $f$. 
These can be easily replicated by using degenerate digital options as building blocks. These in turn can be 
constructed by aggregating $AD$ securities.

In Section \ref{Subsection:path-dependentAD} we have described, by using $AD_{\ua}$ and $AD_{\da}$ derivatives 
step-by-step (Formula \eqref{Eq:ADprice}), how to price a path-dependent $AD$ derivative that pays us $1\n$ at the time of the maturity if a given trajectory occurs.

Since we are now considering a path-independent option, we construct the hedging portfolio using path-independent digital options. The digital option payoff is 
\[f_{digi, K}\ =\ 1_{S_T=K} \]
where $S_T=S_0(1+u)^x(1+d)^{T-x}$; it is irrelevant which particular path the value of the underlying stock follows. Such a digital option can be aggregated from all such path-dependent $AD$ derivatives which follow some trajectory containing exactly $x_0$ `up'-moves. Therefore, 
\begin{equation}\label{eq: Cdigi}
C_{digi, K}(0) \n =\binom{T}{x_0 }q^{x_0} (1-q)^{T-x_{0}} \n.
\end{equation}
Here $K=S_0 (1+u)^{x_0} (1+d)^{T-x_0 }$ and the binomial coefficient $\binom{T}{x_0 }$ is the number of different paths that consist of exactly $x_0$ `up'-moves, i.e. the number of ways how the $x_0$ `up'-moves can be ordered in the paths in question. 

Let us next study a European-style payoff function $f$.
Clearly, the value of a portfolio of $f(K)$-many $C_{digi, K}$ options,  
at time $t=0$ and $t=T$ are
\[
f(K)\binom{T}{x_0 }q^{x_0} (1-q)^{T-x_{0}}\n_{(0)},\quad f(K)1_{S_T=K}\n_{(T)},
\]
respectively. Therefore a general European-style option payoff $f$ can be matched in a simple manner
by a portfolio $\pi$ of digital options in such a way that $f$ and the payoff of the portfolio coincide exactly at maturity $T$.
According to the static hedging principle at time $t=0$ the price of the derivative equals the portfolio value:  
\[
  V_\pi  (0) \n = C_f(0) \n =\sum_{x=0}^{T}f\left(S_0 (1+u)^x (1+d)^{T-x} \right)\binom{T}{x}q^x (1-q)^{T-x} \n
\]
which is essentially the well-known CRR pricing formula \eqref{eq: CRRformula}.

\begin{example}\label{example: }
Let us consider a $2$-step model and suppose that we want to price a European call option with a strike price $K=105\n_{(2)}$, and thus with a payoff $f(\omega)=\max(S_2-105,0)$. Let $r=0.04$ and let the value process of the underlying satisfy $S_0=100$, $u=0.2$, $d=-0.1$, and $p=1/2$, as in Figure \ref{figure:example}.

\begin{figure}[h!]
\begin{center}
\includegraphics[width=6cm]{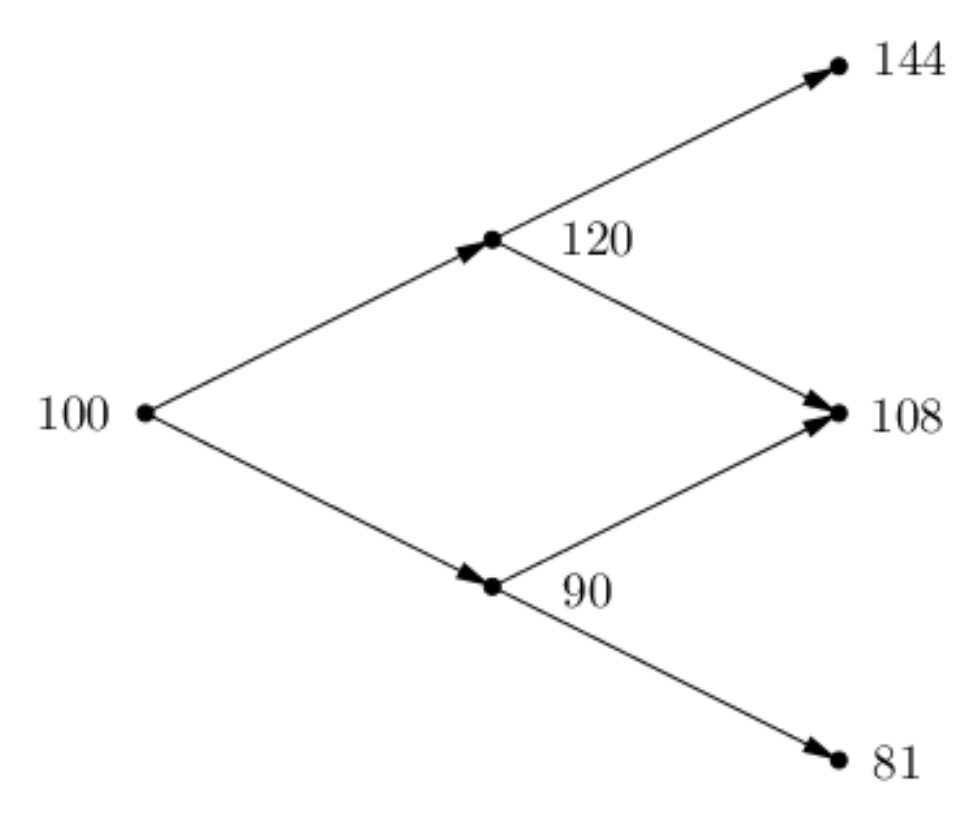}
\caption{Value process of the underlying asset drawn in Log-scale.}
\label{figure:example}
\end{center}
\end{figure}

We shall consider each possible state $S_2$ individually and construct  hedging portfolios for these states by using $AD_\ua$ and $AD_\da$ derivatives.

\begin{itemize}
\item [$1^\circ$] Let us start with the end state having the most obvious hedging strategy, $S_2=S_0(1+d)^2=81\n_{(2)}$. In this case the value of the call is $C_f\n_{(2)}=f(0,0)\n_{(2)}=\max(81-105,0)\n_{(2)}=0 \n_{(2)}$, and since the call is worthless, we need no initial capital to hedge it.

\item [$2^\circ$] Let us then consider the end state $S_2=S_0(1+u)^2$. The only trajectory $\omega$ from $S_0$ to the state $S_2 =144$ is $\omega=(1,1)$. Let us construct the hedging portfolio $\pi$. At the time of the maturity  we want the portfolio to be worth of $f(1,1)\n_{(2)}$. To accomplish this, at time $t=1$ we shall buy $f(1,1)$ shares of $AD_\ua$ derivatives. Each of these will pay us $1\n_{(2)}$ at time $t=2$, so in total our portfolio will have the desired value. These derivatives cost us $f(1,1)q\n_{(1)}$. With similar reasoning, at time $t=0$ we shall buy $f(1,1)q$ shares of $AD_\ua$ derivatives, which will cost us $f(1,1)q^2\n_{(0)}$.

\item [$3^\circ$] Let us consider the remaining state, $S_2=S_0(1+u)(1+d)=108$. Now there are two possible trajectories $\omega_1$ and $\omega_2$ from $S_0$ to $S_2$; these are $\omega_1=(1,0)$ and $\omega_2=(0,1)$. Again, let us construct the hedging portfolio such that at the time of the maturity it has the value $f(1,0)\n_{(2)}=f(0,1)\n_{(2)}$. At time $t=1$, the underlying can satisfy either $S_2=S_1(1+d)$ or $S_2=S_1(1+u)$. In the former case, at time $t=1$ we shall buy $f(1,0)$ shares of $AD_\da$ derivatives, which will cost us $f(1,0)(1-q)\n_{(1)}$. In the latter case, at time $t=1$ we shall buy $f(1,0)$ shares of $AD_\ua$ derivatives, which will cost us $f(1,0)q\n_{(1)}$. 

At time $t=0$ we need to provide for both possibilities for the evolution of the underlying asset, that is, we shall purchase both $f(1,0)(1-q)$ shares of $AD_\ua$ derivatives and $f(1,0)q$ shares of $AD_\da$ derivatives. These will cost us a total of
\[
f(1,0)(1-q)q+f(1,0)q(1-q)\ \n_{(0)}=2f(1,0)q(1-q)\ \n_{(0)}.
\]
\end{itemize}

Finally, we will have the hedging portfolio for the call as the sum of the above, that is
\[
C_f\n_{(0)}=\left(0+f(1,1)q^2+2f(1,0)q(1-q)\right)\n_{(0)}.
\]
After substituting $q=(r-d)/(u-d)=(0.04-(-0.1))/(0.2-(-0.1))=7/15$ and the call's payoffs, we will have the price of the call as
\[
C_f\n_{(0)}=\left(0+39\cdot\left(\frac{7}{15}\right)^2+2\cdot3\cdot\left(\frac{7}{15}\right)\left(\frac{8}{15}\right)\right)\n_{(0)}\approx 9.99 \n_{(0)}
\]
which coincides with the CRR price, see \eqref{eq: CRRformula}.

\end{example}

\section{An alternative route to the CRR formula:\\ Extending the state space and reversing the random walk}
\label{Subsection:extendedCRR}

Let us first extend the state space as
\begin{equation}
\label{eq:statespace}
 \text{States}: \ S_0(1+u)^k(1+d)^l ,\ k,l =0,1,2,\ldots
\end{equation}
so that it includes an infinite number of states. Here $S_0$ is fixed.

We wish again to price a European style derivative with a payoff function $f$. 
By using the static hedging principle seen previously we may accomplish this by constructing a portfolio 
with an end state $K=S_T$:
\[C_f(0) \n =\sum_{x=0}^{T}f(K ) C_{digi,K} (0) \n .\]
Thus it suffices to price each individual $C_{digi,K}$ option separately. Of course, by now we know to expect
the form \eqref{eq: Cdigi}.

\subsection{Backward recursion on degenerate digital option}
\label{Subsection:CRRdetails}
 
Let us begin the construction of this derivative $C_{digi, K}$ at the time of the maturity, $t=T$. At time $t=T$ we wish to receive $1 \n$ if the underlying asset has the value of $K$ and $0 \n$ otherwise.

\begin{figure}[h!]
\begin{center}
\includegraphics[width=11.3cm]{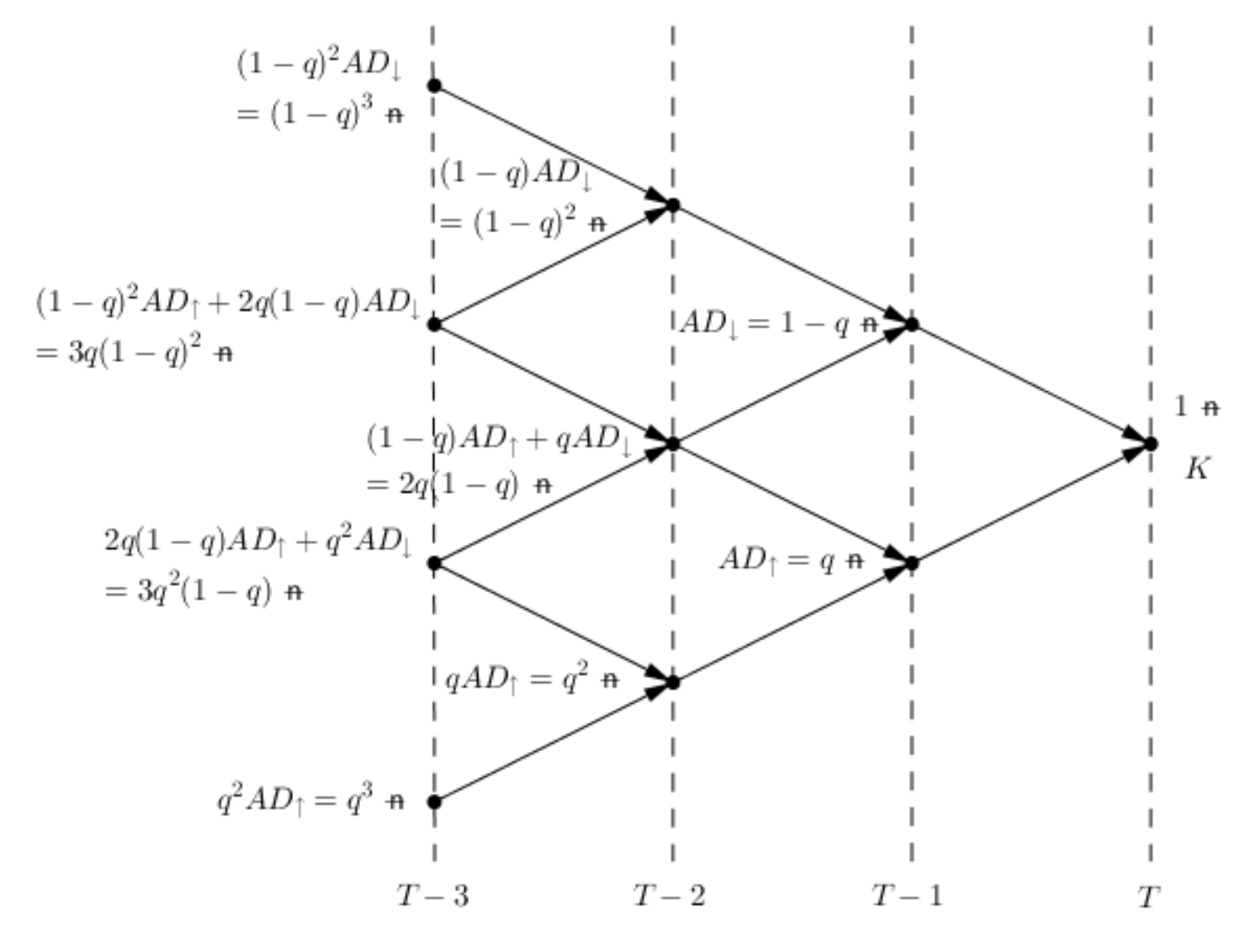}
\caption{First steps of constructing the required degenrate digital option from $AD_{\ua}$ and $AD_{\da}$ derivatives (drawn in Log-scale).}
\label{figure:firststeps}
\end{center}
\end{figure}

{\noindent \emph{At time $t=T-1$}}: There are two possible states of $S_{T-1}$ that enable the $1\n$ payoff at time $T$; these cases are $S_T=S_{T-1}(1+d)$ and  $S_T=S_{T-1}(1+u)$. If $S_T=S_{T-1}(1+d)$, at time $T-1$ we require the wealth to purchase or construct an $AD_{\da}$ derivative which will pay us the desired $1 \n$ at time $T$, i.e. we require $1-q  \n_{(T-1)}$. Respectively, if $S_T=S_{T-1}(1+u)$, at time $T-1$ we require the wealth to obtain an $AD_{\ua}$ derivative which will pay us the desired $1 \n$ at time $T$, i.e. we require the wealth $q \n_{(T-1)}$.
\ \\ 

{\noindent \emph{At time $t=T-2$}}: Similarly, now there are three possible states of $S_{T-2}$
that can enable the $1\n$ payoff at time $T$, via the two states described above. These cases are 
\[S_T=S_{T-2}(1+d)^2,\quad  S_T=S_{T-2}(1+u)^2 ,\quad  S_T=S_{T-2}(1+u)(1+d) .\] 
If $S_T=S_{T-2}(1+d)^2$, at time $T-2$ we require the wealth to buy $1-q$ shares of \ $AD_{\da}$ derivatives (these derivatives will pay us $1 \n$ each at time $T-1$, so at time $T-1$ we will have $(1-q) \n$ which is the amount of wealth that assures obtaining the payoff $1\n$ at time $T$). Thus we require $(1-q)^2 \n$. If $S_T=S_{T-2}(1+u)^2$, with similar reasoning, at time $T-2$ we need to have $q^2 \n$ to assure obtaining the $1 \n$ at time of the maturity. If $S_T=S_{T-2}(1+u)(1+d)=S_{T-2}(1+d)(1+u)$, at time $T-2$ we require the wealth to buy both $1-q$ shares of $AD_{\ua}$ derivatives and $q$ shares of $AD_{\da}$ derivatives. Since we cannot predict which state of nature will occur at time $T-1$, we need to provide for both. Thus we require the wealth $((1-q)q+q(1-q)) \n =2q(1-q) \n$.
 
We can proceed this replicating strategy step by step from the time of the maturity to the beginning, time $t=0$. 
The first steps of constructing the hedging portfolio are represented in Figure \ref{figure:firststeps}.
 
\subsubsection{Backward random walk interpretation} 
 
Note that at each point in the state space we have essentially the same discounted value process; this is represented in Figure \ref{figure:valueprocess}.

\begin{figure}[h!]
\begin{center}
\includegraphics[width=11.3cm]{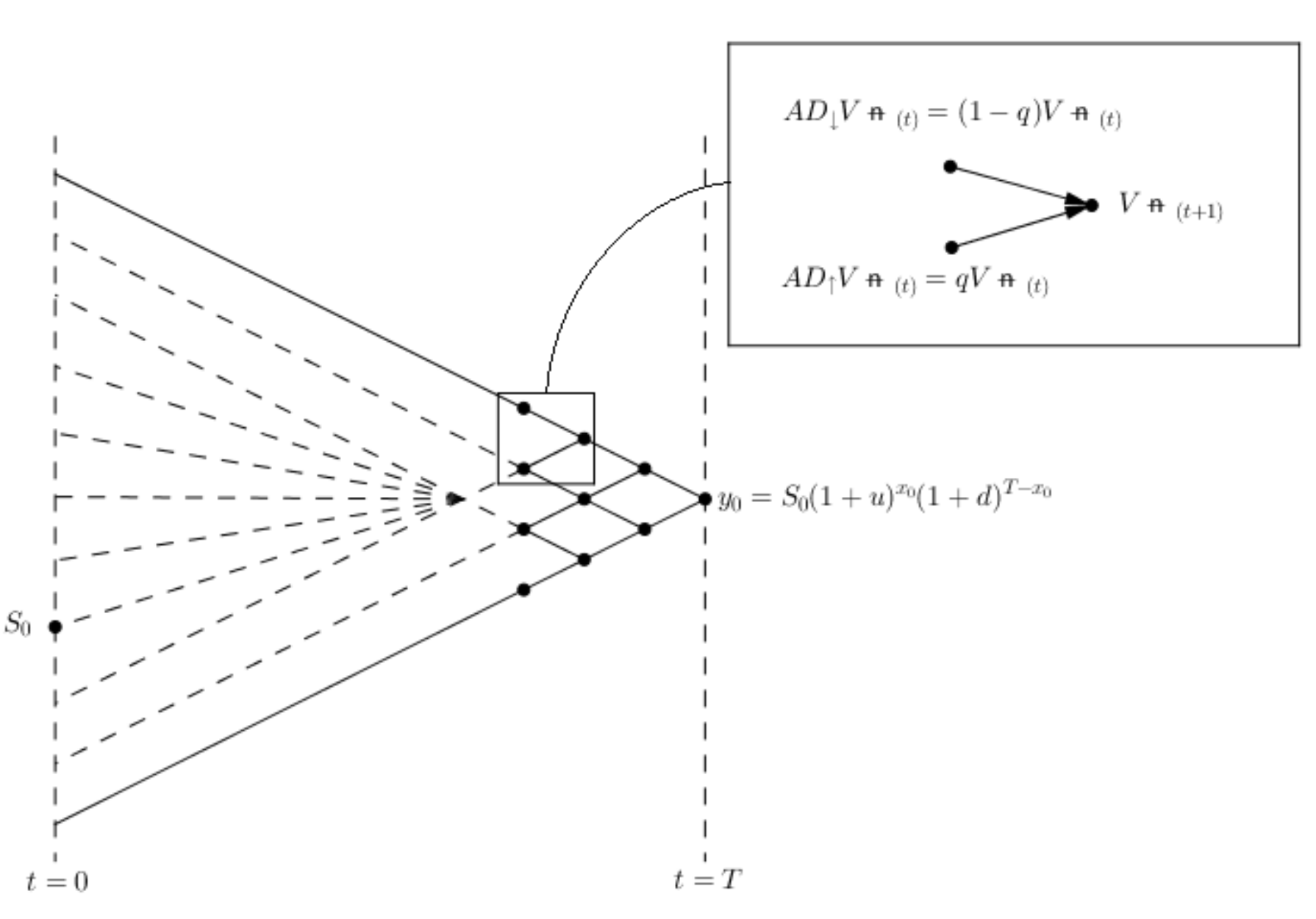}
\caption{Value process of the underlying asset $S$. At each point in the state space, the discounted value process is essentially the same, represented in the ''zoomed-in'' box.}
\label{figure:valueprocess}
\end{center}
\end{figure}


 Let us define a random walk $Y$ (starting from the strike price $K$ of the digital option) as
 \begin{equation*}
 \begin{split}
     Y_0&=y_0 =K\\
     Y_{t+1}&=(1+u)^{\xi_{t+1}}(1+d)^{1-\xi_{t+1}}Y_t
     \end{split}
 \end{equation*}
 where $\xi$ is a biased `coin flip process', i.e. $\xi=(\xi_t)_{t\leq T}$ such that $\xi_t(\omega)=\omega_t$ are independent and identically distributed (i.i.d.) with
 \[
 \mathbf{P}(\xi_t=1)=1-q \hspace{5mm} \text{and} \hspace{5mm} \mathbf{P}(\xi_t=0)=q.
 \]
 Thus, here we consider $q$ as a probability of the event $\xi_t=0$. This random walk can be depicted traveling backward in time as follows.


Now the price $C_{digi, K}(0)\n$ is numerically the probability that $Y$ hits $S_0$,
 \[
 \mathbf{P}(Y_T=S_0)=\mathbf{P}\left(\sum_{t=1}^T (1-\xi_t)=x\right)=\binom{T}{x}q^x(1-q)^{T-x}.
 \]
Here $\xi_t$ is binomially distributed and therefore $1-\xi_t$ is also binomially distributed. The required probability is obtained by the probability mass function of a binomially distributed random variable. Also note that $1-\xi_t$ satisfies 
\[\mathbf{P}((1-\xi_t)=0)=\mathbf{P}(\xi_t=1)=1-q\quad \text{and}\quad \mathbf{P}((1-\xi_t)=1)=\mathbf{P}(\xi_t=0)=q.\]
The reason we require here some form of extension of the state space is to enable the branching of the backward random walk.

\section{Some further remarks}

\subsection{An invariance property in the extended CRR model}

In the extended CRR model in Section \ref{Subsection:extendedCRR} for each time $t$ 
the probabilities of the backward process sum to the unity over the states, cf. Figure \ref{figure:prop1}. In other words, the replication $\n$ values 
of a $C_{digi ,K}$ option over the states do not depend on time, as the aggregate is $1$. 

\begin{figure}[h!]
\begin{center}
\includegraphics[width=7cm]{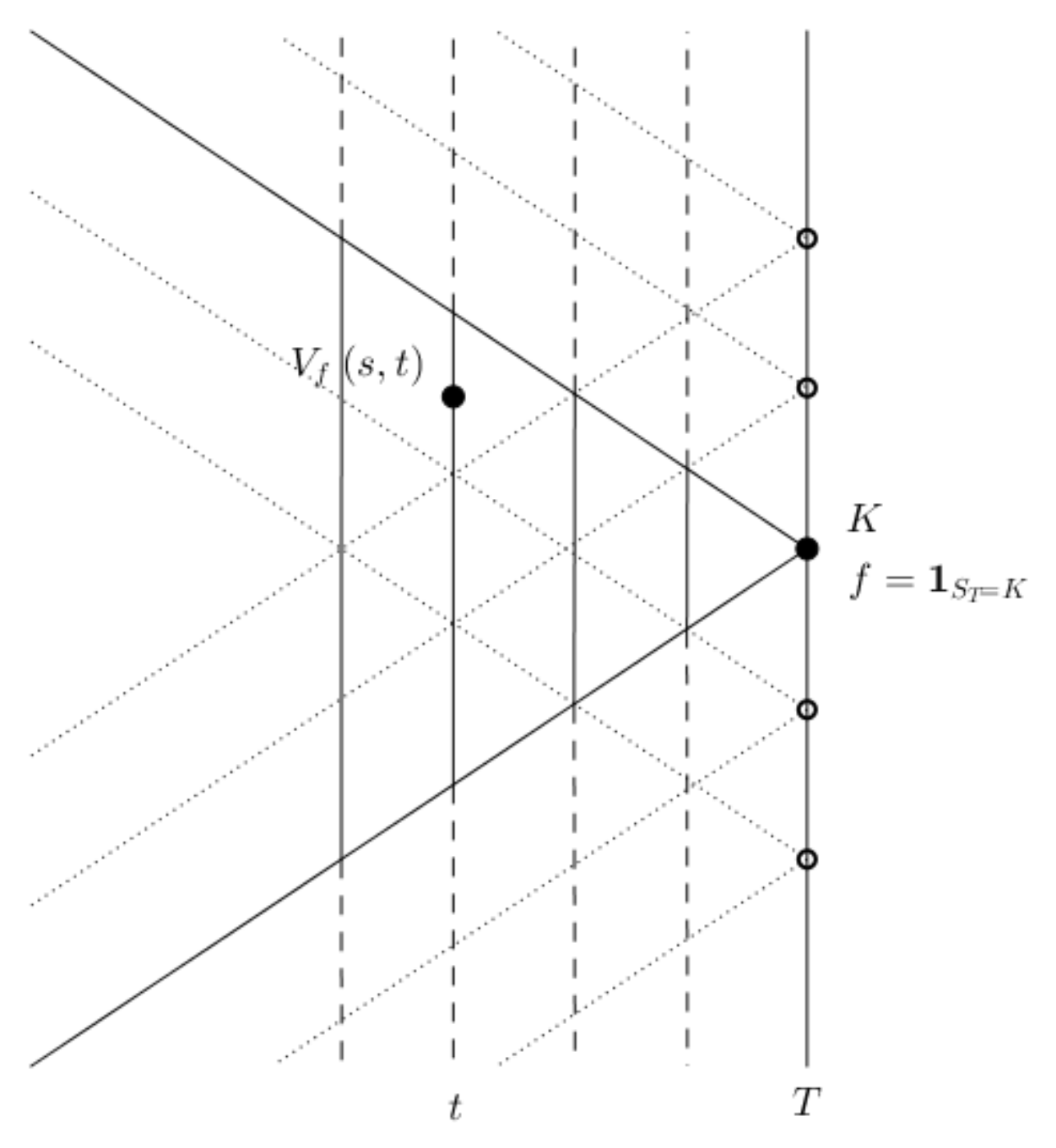}
\caption{The value process corresponding to the strike price $K$.}
\label{figure:prop1}
\end{center}
\end{figure}

Taking this observation a bit further, if we statically hedge European-style derivatives by aggregating degenerate digital options, we also aggregate with respective weights the probabilities of the processes starting from different states at the maturity.

Recall the extended state space in \eqref{eq:statespace}, and let $V_f(s,t)$ be the payoff's value corresponding the state $s$ and the time $t$. Let us write 
\[
V_f(s,t)=\sum_K V_f^{(K)}(s,t),
\]
where $V_f^{(K)}(s,t)$ is the value process corresponding to a degenerate digital option with a payoff $f(K)\cdot 1_{S_T=K}$.

\begin{proposition}\label{prop: aggregate}
Consider the model in Section \ref{Subsection:extendedCRR} and let $f$ be the payoff of a European-style derivative such that 
\[
\sum_K |f(K)|<\infty\ \text{or}\ f\geq 0 .
\]
Then the value process of the replication strategy for $C_f$ satisfies
\[
\sum_{s} V_f (s,t) \n_{(t)} =\sum_K f(K) \n_{(T)}\ \ \text{for\ all}\ \ t\leq T .
\]
\end{proposition}

\begin{proof}
\[
\sum_{s} V_f (s,t) \n_{(t)} =\sum_s\sum_K V_f^{(K)}(s,t) \n_{(t)} = \sum_K \sum_s V_f^{(K)}(s,t) \n_{(t)} 
=\sum_K f(K) \n_{(T)}.
\]
According to the assumptions we may change the order of summations in the middle equality. The last equality holds because if the payoff of the security is $1$, then the possible values at time $t$, which coincide with the probabilities, sum to $1$ (see Figure \ref{figure:prop1}). Similarly, since the degenerate digital option has the payoff $f(K)$, the possible values at time $t$ sum to $1\cdot f(K)$.
\end{proof}

The standard CRR model fails the above property. Namely, consider at time $t=0$ the one and only state $S_0$ and our 
derivative in this example is trivially the risk-free bond with maturity at $T$. Then the $\n$-value of the bond is $1$ but the sum $\sum_K f(K) = \sum_K 1$ becomes large, it is the number of all possible states at time $T$. 

On the other hand, the BSM model has the similar property, namely
\[\int_{-\infty}^\infty  \int_{-\infty}^\infty f(e^y) \varphi(y-x)\ dy\ dx =  \int_{-\infty}^\infty f(e^y) 
\int_{-\infty}^\infty  \varphi(y-x)\ dx\ dy\ =\int_{-\infty}^\infty f(e^y)\ dy\]
where we used $y=\ln S_T$, $x=\ln S_0$, $\varphi$ is the risk-neutral density function of the BSM model.
Recall that the distribution is $\mathcal{N}((r- \frac{1}{2}\sigma^2)T, \sigma^2 T)$ with the relevant model parameters in place.

In what follows, we will provide an example of a situation where this invariance property has an interesting implication.  Let us consider a digital option with a payoff function
\[
f(\omega)=1_{K_1\leq S_T \leq K_2},
\]
i.e. the digital option pays $1 \n$  if the underlying hits the interval $[K_1,K_2]$, and $0\n$ otherwise. We assume here that the strike $K$ can have only discrete values.
As stated previously, the values $V_f^{(K)}(s,t)$ corresponding to the strike $K$ sum to $1$ at any time $t\leq T$. Applying this property, by summing all the possible values $V_f^{(K)}(s,t)$ corresponding to all strikes $K$, such that $K_1\leq K\leq K_2$, at time $t$, we can conclude that the sum must equal the number of possible strikes in the interval $[K_1,K_2]$.
%

\renewcommand{\t}{\ \!\! \vartriangle\! (t)\ \!\! } 
\subsection{Why the trend term $\mu$ does not appear in the BSM prices?}
The fact that the trend term $\mu$ does not affect prices in the BSM pricing appears rather counterintuitive.
There are some anecdotes on how the pricing formulas were suspected before the seminal paper of Black and Scholes was published and even the authors first doubted their findings. 

We will discuss here the irrelevance of $\mu$ in the BSM model, as seen from the lattice model asymptotics along vanishing step size. 
The $\mu$ parameter cannot be excluded in the binomial framework in the formation of the risk-neutral probabilities $q$. 
On the other hand, the effect of $\mu$ should vanish as the time-scale is refined and the binomial models converge to a BSM model (in a suitable sense). Next we will analyze the \emph{speed of convergence} of the risk-neutral variance of the underlying binomial process. Recall that the asymptotic log-Normal state-price density
can be recovered in principle by normal approximation of the binomial distribution  
from the risk-neutral expectation (see \eqref{eq: R} below) and variance of the jumps, since they are i.i.d.

To this end we will fix the following dependence of returns on the parameters:
\[U = e^{\mu \t + \sqrt{\t}\sigma},\quad D = e^{\mu \t - \sqrt{\t}\sigma},\quad R=e^{r \t}.\]
Here we have time step $\t$ and 
the usual BSM model parameters, $\mu>0$ is the trend of the underlying and $\sigma>0$ the standard deviation or volatility term and 
$r>0$ the short rate. Some reasonable values could be $\mu=0.1 ,\ \sigma=0.2$ and $r=0.04$.
Mimicing the BSM model, the lattice model of the underlying asset is 
\[S_{t+1} = S_t  e^{\mu \t + \sqrt{\t}\sigma \theta_t} \]
where $\theta_t$ are i.i.d. random variables with $\P(\theta=1)=\P(\theta=-1)=\frac{1}{2}$.
      
Here we will use the following risk-neutral single-step probabilities as above:
\[q=\frac{R-D}{U-D},\quad (1-q)=\frac{U-R}{U-D}.\]
Then, by simple algebra we obtain the following identities:
\begin{equation}\label{eq: R}
q U + (1-q) D =\frac{R-D}{U-D}U + \frac{U-R}{U-D}D  =R,
\end{equation}
\begin{equation}\label{eq: U2}
\frac{R-D}{U-D}U^2 + \frac{U-R}{U-D}D^2  = RU + RD -UD.
\end{equation}
Equation \eqref{eq: R} says that in the risk-neutral world the expected return of the underlying asset is the risk-free return and in particular
does not depend on $\mu$.

The risk-neutral single-step variance of the asset return is
\[\frac{R-D}{U-D}(U-R)^2 + \frac{U-R}{U-D}(D-R)^2 \] 
and the risk-neutral variance of $\log(S_T / S_0 )$ for small $\t$ is approximately
\begin{equation}\label{eq: TVar}
\frac{T}{\t} \left(\frac{R-D}{U-D}(U-R)^2 + \frac{U-R}{U-D}(D-R)^2 \right).
\end{equation}
where the $\frac{T}{\t} $ is the total number of steps in the time span. Indeed, we apply the fact that 
$\mathrm{Var}\ e^{\t \theta} \approx \t\ \mathrm{Var}\ \theta$ for small $\t$.

This reads
\begin{eqnarray*}
\mathrm{Var}_\mathbf{Q}\ \log(S_T / S_0 ) &\approx& \frac{T}{\t} \left(\frac{R-D}{U-D} (U^2 - 2RU + R^2 ) +  \frac{U-R}{U-D} (D^2 - 2RD + R^2 )  \right) \\
&=& \frac{T}{\t} \left(\frac{R-D}{U-D} U^2  +  \frac{U-R}{U-D} D^2  - R^2   \right) \\
&=& \frac{T}{\t} (RU + RD -UD - R^2 )
\end{eqnarray*}
where the second equality follows from \eqref{eq: R} and by thinking of the risk-neutral expectations, and the last one from \eqref{eq: U2}.

To analyse the contribution of $\mu$ on the risk-neutral variance, we shall analyse the Taylor expansion of the above terms:
\begin{eqnarray*}
&&RU + RD - UD - R^2\\ 
&=&  e^{\mu \t + \sqrt{\t}\sigma + r\t } + e^{\mu \t - \sqrt{\t}\sigma + r\t } - 
e^{2\mu \t} - e^{2r\t}\\
&=& 1+ (\mu \t + \sqrt{\t}\sigma + r\t ) +\frac{1}{2} (\mu \t + \sqrt{\t}\sigma + r\t )^2\\ 
&+& \frac{1}{6} (\mu \t + \sqrt{\t}\sigma + r\t )^3  \ldots\\
&+& 1+ (\mu \t - \sqrt{\t}\sigma + r\t ) +\frac{1}{2} (\mu \t - \sqrt{\t}\sigma + r\t )^2\\ 
&+& \frac{1}{6} (\mu \t - \sqrt{\t}\sigma + r\t )^3 \ldots\\
&-&  1 - 2\mu \t  - \frac{1}{2}( 2\mu \t )^2  -\frac{1}{6}( 2\mu \t )^3  - \ldots \\
&-&  1 - 2r \t - \frac{1}{2}( 2r \t )^2  - \frac{1}{6}( 2r \t )^3 \ldots \\
&=& \t \sigma^2  + (\t)^2 (\mu + r) \sigma^2 + \ldots
\end{eqnarray*}

Consequently, we obtain an approximation for the risk-neutral variance:
\begin{eqnarray*}
\mathrm{Var}_\mathbf{Q}\ \log(S_T / S_0 )&\approx& \frac{T}{\t}\left(\t \sigma^2 + (\t)^2 (\mu + r) \sigma^2 \right)\\
&=& T \sigma^2 \left(1  + (\mu + r) \t\right) .
\end{eqnarray*}
From this form we immediately see the BSM model variance which arises asymptotically as $\t \to 0$:
\begin{equation}\label{eq: VarBSM}
\mathrm{Var}^{\mathrm{BSM}}_\mathbf{Q}\ \log(S_T / S_0 ) = T \sigma^2.
\end{equation}

The conclusion here is that, in an annual binomial model, if the time step is $1$ day ($\t=1/365$) then the effect of $\mu$ and $r$ on the underlying asset's risk-neutral variance is negligible.
According to \eqref{eq: R} and \eqref{eq: VarBSM} parameter $\mu$ does not appear in the $\mathbf{Q}$-distribution which is used in pricing options.

To conclude, we will comment heuristically the vanishing effect of $\mu$ on the risk-neutral probability. 
This is a bit of a paradox and the above calculations give only dim insight into what is `really' happening here. 

Fixing small $\t$ we have 
\[q=\frac{R-D}{U-D} \approx \frac{r\t -(\mu\t - \sqrt{\t}\sigma)}{\sqrt{\t}2\sigma}.\]

Thus, the effect of changes in $\mu$ on $q$ is 
\begin{equation}\label{eq: vartriq}
\vartriangle q = - \frac{\sqrt{\t} \vartriangle \mu}{2\sigma}.
\end{equation}

This means that as $\mu$ increases the risk-neutral probability mass shifts down in the tree. 

So, how can the the risk-neutral probability measure value $\mathbf{Q}(S_T)$ be asymptotically invariant of $\mu$?
Clearly the above sensitivity \eqref{eq: vartriq} decreases for small time scale. 

Also, note that the risk-neutral density concerns explicitly the value of $S_T$ and not the number of up jumps. Recall that in \eqref{eq: CRRformula} the value of $S_T$ appears rather indirectly as the number of jumps $x$, so let us write $x=x(S_T)$, the number of up jumps required for a given terminal asset price $S_T$. 

Changing a down jump to an up jump results in $\sqrt{\t} 2 \sigma$ increase in the log-price of the asset.
On the other hand, changing $\mu$ affects uniformly every time-step of the model, so the corresponding change is 
\[\vartriangle \log (S_T) = T \vartriangle \mu .\]
This means that if we wish to counteract the increase of $\mu$ by changing the number of up steps, the required adjustment is
\[\vartriangle x (S_T) =- \frac{T \vartriangle \mu}{\sqrt{\t} 2 \sigma}.\]
Consequently, increasing $\mu$ shifts the risk-neutral probability mass down the tree but it simultaneously shifts down the end node in the tree, corresponding to a fixed value $S_T$.

\section{Conclusions}
In this paper, we provide a leaner, not as technical, proof for the Cox-Ross-Rubinstein pricing formula. We have made an effort to simplify the proof and make it pedagogically more approachable by emphasising the financial intuition behind the pricing formula.

The fundamental idea of our proof is, by using the static hedging argument, to construct a replicating portfolio using Arrow-Debreu securities and digital options. We start this construction from the time of the maturity, and proceed backwards to the time $t=0$. In order to enable the backward recursion, we extend the state space. In this extended CRR model, there exists an interesting invariance property: at each time $t\leq T$, the sum of all possible values of the stock corresponds to the sum of all possible payoffs at maturity.
We show one example where this invariance property can be used in the analysis of financial derivatives. In addition to our example, the invariance property can have various applications in financial mathematics. 

At the end of the paper, we discuss the paradox of the trend parameter $\mu$ not affecting the prices of derivatives. We provide justification for this well-known fact by showing that the risk-neutral density $\mathbf{Q}$, that appears in the pricing formula, is independent of the parameter $\mu$.

\end{document}